\documentclass[a4paper,11pt]{article}

\usepackage[in]{fullpage}
\usepackage{hyperref}
\usepackage[utf8]{inputenc}
\usepackage{amsfonts,amssymb,amsthm,bbm, mathtools}
\usepackage{thm-restate}
\usepackage{footnote}
\usepackage{authblk}

\hypersetup{
colorlinks = true,
citecolor = blue
}

\newtheorem{theorem}{Theorem}
\newtheorem{lemma}[theorem]{Lemma}

\newtheorem{proposition}[theorem]{Proposition}
\theoremstyle{definition}

\newcommand{\clB}{\mathcal B}

\newcommand{\clT}{\mathcal T}

\newcommand{\GTH}{\textsc{Gth}}
\newcommand{\OSP}{\textsc{Osp}}
\newcommand{\IND}{\textsc{Ind}}
\newcommand{\FPP}{\textsc{Fpp}}
\newcommand{\OR}{\textsc{Or}}
\newcommand{\AND}{\textsc{And}}
\newcommand{\Tribes}{\textsc{Tribes}}

\DeclareMathOperator{\bs}{bs}
\DeclareMathOperator{\fbs}{fbs}
\DeclareMathOperator{\FC}{FC}
\DeclareMathOperator{\C}{C}
\DeclareMathOperator{\RC}{RC}
\DeclareMathOperator{\QC}{QC}
\DeclareMathOperator{\CA}{CRA}
\DeclareMathOperator{\MM}{CMM}
\DeclareMathOperator{\KA}{CKA}
\DeclareMathOperator{\WA}{CWA}
\DeclareMathOperator{\prt}{prt}
\DeclareMathOperator{\RS}{RS}
\DeclareMathOperator{\D}{D}
\DeclareMathOperator{\R}{R}
\DeclareMathOperator{\Q}{Q}
\DeclareMathOperator{\RL}{RL}
\DeclareMathOperator{\LL}{LL}

\title{All Classical Adversary Methods are Equivalent \\ for Total Functions \thanks{This work is supported by the ERC Advanced Grant MQC and Latvian State Research Programme NexIT Project No. 1.}}

\date{}

\begin{document}

\renewcommand\Affilfont{\small}

\author{Andris Ambainis}
\author{Martins Kokainis}
\author{Krišjānis Prūsis}
\author{\authorcr Jevgēnijs Vihrovs}
\author{Aleksejs Zajakins}
\affil{Centre for Quantum Computer Science, Faculty of Computing, \authorcr University of Latvia, Rai\c{n}a 19, Riga, Latvia, LV-1586.}

\maketitle

\begin{abstract}
We show that all known classical adversary lower bounds on randomized query complexity are equivalent for total functions, and are equal to the fractional block sensitivity $\fbs(f)$.
That includes the Kolmogorov complexity bound of Laplante and Magniez and the earlier relational adversary bound of Aaronson.
This equivalence also implies that for total functions, the relational adversary is equivalent to a simpler lower bound, which we call rank-1 relational adversary.
For partial functions, we show unbounded separations between $\fbs(f)$ and other adversary bounds, as well as between the adversary bounds themselves.

We also show that, for partial functions, fractional block sensitivity cannot give lower bounds larger than $\sqrt{n \cdot \bs(f)}$, where $n$ is the number of variables and $\bs(f)$ is the block sensitivity.
Then we exhibit a partial function $f$ that matches this upper bound, $\fbs(f) = \Omega(\sqrt{n \cdot \bs(f)})$.
\end{abstract}

\section{Introduction}

Query complexity of functions is one of the simplest and most useful models of computation.
It is used to show lower bounds on the amount of time required to solve a computational task, and to compare the capabilities of the quantum, randomized and deterministic models of computation.
Thus providing lower bounds in the query model is essential in understanding the complexity of computational problems.

In the query model, an algorithm has to compute a function $f : S \to H$, given a string $x$ from $S \subseteq G^n$, where $G$ and $H$ are finite alphabets.
With a single query, it can provide the oracle with an index $i \in [n]$ and receive back the value $x_i$.
After a number of queries (possibly, adaptive), the algorithm must compute $f(x)$.
The cost of the computation is the number of queries made by the algorithm.

The query complexity of a function $f$ in the deterministic setting is denoted by $\D(f)$ and is also called the decision tree complexity.
The two-sided bounded-error randomized and quantum query complexities are denoted by $\R(f)$ and $\Q(f)$, respectively (which means that given any input, the algorithm must produce a correct answer with probability at least 2/3).
For a comprehensive survey on the power of these models, see \cite{BdW02}, and for the state-of-the-art relationships between them, see \cite{ABDK16}.

In this work, we investigate the relation among a certain set of lower bound techniques on $\R(f)$, called the classical adversary methods, and how they connect to other well-known lower bounds on the randomized query complexity.

\subsection{Known Lower Bounds}
One of the first general lower bound methods on randomized query complexity is Yao's minimax principle, which states that it is sufficient to exhibit a hard distribution on the inputs and lower bound the complexity of any deterministic algorithm under such distribution \cite{Yao77}. Yao's minimax principle is known to be optimal for any function but involves a hard-to-describe and hard-to-compute quantity (the complexity of the best deterministic algorithm under some distribution).

More concrete randomized lower bounds are block sensitivity $\bs(f)$ \cite{Nis89} and the approximate degree of the polynomial representing the function $\widetilde \deg(f)$ \cite{NS94}  introduced by Nisan and Szegedy.
Afterwards, Aaronson extended the notion of the certificate complexity $\C(f)$ (a deterministic lower bound) to the randomized setting by introducing randomized certificate complexity $\RC(f)$ \cite{Aar08}.
Following this result, both Tal and Gilmer, Saks and Srinivasan independently discovered the fractional block sensitivity $\fbs(f)$ lower bound \cite{Tal13,GSS16}, which is equal to the fractional certificate complexity $\FC(f)$ measure, as respective dual linear programs.
Since these measures are relaxations of block sensitivity and certificate complexity if written as integer programs, they satisfy the following hierarchy:
\begin{equation*}
\bs(f) \leq \fbs(f) = \FC(f) \leq \C(f).
\end{equation*}
Perhaps surprisingly, fractional block sensitivity turned out to be equivalent to randomized certificate complexity, $\fbs(f) = \Theta(\RC(f))$.
Approximate degree and fractional block sensitivity are incomparable in general, but it has been shown that $\fbs(f) \leq \widetilde\deg(f)^2$ \cite{KT16} and $\widetilde\deg(f) \leq \bs(f)^3 \leq \fbs(f)^3$ \cite{Nis89,BBCMdW01}.

Currently one of the strongest lower bounds is the partition bound $\prt(f)$ of Jain and Klauck \cite{JK10}, which is larger than all of the above mentioned randomized lower bounds (even the approximate degree), and the classical adversary methods listed below.
Its power is illustrated by the $\Tribes_n$ function (an $\AND$ of $\sqrt n$ $\OR$s on $\sqrt n$ variables), where it gives a tight $\Omega(n)$ lower bound, while all of the other lower bounds give only $O(\sqrt n)$.
The quantum query complexity $\Q(f)$ is also a powerful lower bound on $\R(f)$, as it is incomparable with $\prt(f)$ \cite{AKK16}.
Recently, Ben-David and Kothari introduced the randomized sabotage complexity $\RS(f)$ lower bound, which can be even larger than $\prt(f)$ and $\Q(f)$ for some functions \cite{BDK16}, and so far no examples are known where it is smaller.

In a separate line of research, Ambainis gave a versatile quantum adversary lower bound method with a wide range of applications \cite{Amb00}.
Since then, many generalizations of the quantum adversary method have been introduced (see \cite{SS06} for a list of known quantum adversary bounds).
Several of these formulations have been lifted back to the randomized setting.
Aaronson proved a classical analogue of Ambainis' relational adversary bound and used it to provide a lower bound for the local search problem \cite{Aar06}.
Laplante and Magniez introduced the Kolmogorov complexity adversary bound for both quantum and classical settings and showed that it subsumes many other adversary techniques. \cite{LM04}.
They also gave a classical variation of Ambainis' adversary bound in a different way than Aaronson.
Some of the other adversary methods like spectral adversary have not been generalized back to the randomized setting.

While some relations between the adversary bounds had been known before, Špalek and Szegedy proved that practically all known quantum adversary methods are in fact equivalent \cite{SS06} (this excludes the general quantum adversary bound, which gives an exact estimate on quantum query complexity for all Boolean functions \cite{HTS07, Rei09}).
This result cannot be immediately generalized to the classical setting, as the equivalence follows through the spectral adversary which has no classical analogue.
They also showed that the quantum adversary cannot give lower bounds better than a certain ``certificate complexity barrier''.
Recently, Kulkarni and Tal strenghtened the barrier using fractional certificate complexity.
Specifically, for any Boolean function $f$ the quantum adversary is at most $\sqrt{\FC^0(f)\FC^1(f)}$, if $f$ is total, and at most $2\sqrt{n\cdot \min\{\FC^0(f),\FC^1(f)\}}$, if $f$ is partial \cite{KT16}.\footnote{Here, $\FC^0(f)$ and $\FC^1(f)$ stand for the maximum fractional certificate complexity over negative and positive inputs, respectively.}

With the advances on the quantum adversary front, one could hope for a similar equivalence result to also hold for the classical adversary bounds.
Some relations are known: Laplante and Magniez have shown that the Kolmogorov complexity lower bound is at least as strong as Aaronson's relational and Ambainis' weighted adversary bounds \cite{LM04}.
Jain and Klauck have noted that the minimax over probability distributions adversary bound is at most $\C(f)$ for total functions \cite{JK10}.
In general, the relationships among the classical adversary bounds until this point remained unclear.

\subsection{Our Results}

Our main result shows that the known classical adversary bounds are all equivalent for total functions.
That includes Aaronson's relational adversary bound $\CA(f)$, Ambainis' weighted adversary bound $\WA(f)$, the Kolmogorov complexity adversary bound $\KA(f)$ and the minimax over probability distributions adversary bound $\MM(f)$.
Surprisingly, they are equivalent to the fractional block sensitivity $\fbs(f)$.

We also add to this list a certain restricted version of the relational adversary bound.
More specifically, we require that the relation matrix between the inputs has rank 1, and denote this (seemingly weaker) lower bound by $\CA_1(f)$.
Thus for total functions $\CA(f) = \Theta \left(\CA_1(f)\right)$, where the latter is much easier to calculate for Boolean functions.

All this shows that $\fbs(f)$ is a fundamental lower bound measure for total functions with many different formulations, including the previously known $\FC(f)$ and $\RC(f)$.
Another interesting corollary is that since the quantum certificate complexity $\QC(f) = \Theta(\sqrt{\RC(f)})$ is a lower bound on the quantum query complexity \cite{Aar08}, we have that by taking the square root of any of the adversary bounds above, we obtain a quantum lower bound for total functions.

Along the way, for partial functions we show the equivalence between $\CA(f)$ and $\WA(f)$, and also between $\KA(f)$ and $\MM(f)$.
In the case of partial functions, $\fbs(f)$ becomes weaker than all these adversary methods.
In particular, we show an example of a function where each of these adversary methods gives an $\Omega(n)$ lower bound, while fractional block sensitivity is $O(1)$.
We also show that $\CA(f)$ and $\MM(f)$ are not equivalent for partial functions, as there exists an example where $\CA(f)$ is constant, but $\MM(f) = \Theta(\log n)$.
Finally, we show a function such that $\CA_1(f) = O(\sqrt n)$, but $\CA(f) = \Omega(n)$.

We also show a ``block sensitivity'' barrier for fractional block sensitivity.
Namely, for any partial function $f$, the fractional block sensitivity is at most $\sqrt{n\cdot\bs(f)}$.
Note that the adversary bounds do not bear this limitation, as witnessed by the aforementioned example.
This result is tight, as we exhibit a partial function that matches this upper bound.

Even though our results are similar to the quantum case in \cite{SS06} in spirit, the proof methods are different.

\section{Preliminaries}

In this section we define the complexity measures we are going to work with in the paper.
In the following definitions and the rest of the paper consider $f$ to be a partial function $f : S \to H$ with domain $S \subseteq G^n$, where $G, H$ are some finite alphabets and $n$ is the length of the input string.
Throughout the paper we assume that $f$ is not constant.

\paragraph{Block Sensitivity.}

For $x \in S$, a subset of indices $B \subseteq [n]$ is a \emph{sensitive block} of $x$ if there exists a $y$ such that $f(x) \neq f(y)$ and $ B=\{ i \mid x_i \neq y_i \}$.
The \emph{block sensitivity} $\bs(f, x)$ of $f$ on $x$ is the maximum number $k$ of disjoint subsets $B_1, \ldots, B_k \subseteq [n]$ such that $B_i$ is a sensitive block of $x$ for each $i \in [k]$.
The block sensitivity of $f$ is defined as $\bs(f) = \max_{x \in S} \bs(f, x)$.

Let $\clB =\left\{B \mid \exists y : f(x) \neq f(y) \text{ and }  B=\{ i \mid x_i \neq y_i \}\right\}$ be the set of sensitive blocks of $x$.
The \emph{fractional block sensitivity} $\fbs(f, x)$ of $f$ on $x$ is defined as the optimal value of the following linear program:
\begin{align*}
\text{maximize } \sum_{B \in \clB} w_x(B) \hspace{1cm} \text{ subject to } \quad & \forall i \in [n]: \sum_{B \in \clB \atop i \in B} w_x(B) \leq 1.
\end{align*}
Here $w_x \in [0;1]^{|\clB|}$.
The fractional block sensitivity of $f$ is defined as $\fbs(f) = \max_{x \in S} \fbs(f, x)$.

When the weights are taken as either 0 or 1, the optimal solution to the corresponding integer program is equal to $\bs(f, x)$.
Hence $\fbs(f,x)$ is a relaxation of $\bs(f,x)$, and we have $\bs(f,x) \leq \fbs(f,x)$.

\paragraph{Certificate complexity.}

An \emph{assignment} is a map $A:\{1, \ldots, n\} \rightarrow G \cup \{*\}$.
Informally, the elements of $G$ are the values fixed by the assignment and * is a wildcard symbol that can be any letter of $G$.
A string $x \in S$ is said to be consistent with $A$ if for all $i \in [n]$ such that $A(i) \neq *$, we have $x_i = A(i)$.
The length of $A$ is the number of positions that $A$ fixes to a letter of $G$.

For an $h \in H$, an $h$-certificate for $f$ is an assignment $A$ such that for all strings $x \in A$ we have $f(x) = h$.
The \emph{certificate complexity} $\C(f, x)$ of $f$ on $x$ is the size of the shortest $f(x)$-certificate that $x$ is consistent with.
The certificate complexity of $f$ is defined as $\C(f) = \max_{x \in S} \C(f, x)$.

The \emph{fractional certificate complexity} $\FC(f, x)$ of $f$ on $x \in S$ is defined as the optimal value of the following linear program:
\begin{equation*}
\text{minimize } \sum_{i \in [n]} v_x(i) \hspace{1cm} \text{ subject to } \quad \forall y \in S \text{ s.t. } f(x) \neq f(y): \sum_{i : x_i \neq y_i} v_x(i) \geq 1.
\end{equation*}
Here $v_x \in [0;1]^n$ for each $x \in S$.
The fractional certificate complexity of $f$ is defined as $\FC(f) = \max_{x \in S} \FC(f, x)$.

When the weights are taken as either 0 or 1, the optimal solution to the corresponding integer program is equal to $\C(f, x)$.
Hence $\FC(f,x)$ is a relaxation of $\C(f,x)$, and we have $\FC(f,x) \leq \C(f,x)$.

It has been shown that $\fbs(f,x)$ and $\FC(f,x)$ are dual linear programs, hence their optimal values are equal, $\fbs(f,x) = \FC(f,x)$.
As an immediate corollary, $\fbs(f) = \FC(f)$.

\paragraph{One-sided measures.}

For Boolean functions with $H= \{0,1\}$, for each measure $M$ from $\bs(f), \fbs(f), \FC(f), \C(f)$ and a Boolean value $b \in \{0, 1\}$, define the corresponding one-sided measure as \begin{equation*}
M^b(f) = \max_{x \in f^{-1}(b)} M(f,x).
\end{equation*}
According to the earlier definitions, we then have $M(f) = \max\{M^0(f), M^1(f)\}$.
These one-sided measures are useful when, for example, working with compositions of $\OR$ with some Boolean function.

\paragraph{Kolmogorov complexity.}

A set of strings $\mathcal S \subset \{0, 1\}^*$ is called \emph{prefix-free} if there are no two strings in $\mathcal S$ such that one is a proper prefix of the other.
Equivalently we can think of the strings as programs for the Turing machine.
Let $M$ be a universal Turing machine and fix a prefix-free set $\mathcal S$.
The prefix-free \emph{Kolmogorov complexity} of $x$ given $y$, is defined as the length of the shortest program from $\mathcal S$ that prints $x$ when given $y$:
\begin{equation*}
K(x | y) = \min \{|P| \mid P \in \mathcal S, M(P, y) = x\}.
\end{equation*}
For a detailed introduction on Kolmogorov complexity, we refer the reader to \cite{MV08}.

\section{Classical Adversary Bounds}

Let $f : S \to H$ be a function, where $S \subseteq G^n$.
The following are all known to be lower bounds on bounded-error randomized query complexity.

\paragraph{Relational adversary bound \cite{Aar06}.}
Let $R : S \times S \to \mathbb R_{\geq 0}$ be a real-valued function such that $R(x,y)=R(y,x)$ for all $x,y \in S$ and $R(x,y)=0$ whenever $f(x)=f(y)$.  Then for $x \in S$ and an index $i$, let\footnote{We take the reciprocals of the expressions, compared to Aaronson's definition.}
\begin{equation*}
\theta(x, i) = \frac{\sum_{y \in S} R(x, y)}{\sum_{y \in S : x_i \neq y_i} R(x, y)},
\end{equation*}
%
where $\theta(x,i)$ is undefined if the denominator is 0. Denote\footnote{One can show that there exist optimal solutions for $R$, thus we can maximize over $R$ instead of taking the supremum.}
\begin{equation*}
\CA(f) = \max_R \min_{x, y \in S, i \in [n] : \atop R(x, y) > 0, x_i \neq y_i} \max\{\theta(x, i), \theta(y, i)\}.
\end{equation*}

\paragraph{Rank-1 relational adversary bound.}
We introduce the following restriction of the relational adversary bound.
Let $R'$ be any $|S| \times |S|$ matrix of rank 1, such that:
\begin{itemize}
\item There exist $u, v : S \to \mathbb R_{\geq 0}$ such that $R'(x,y) = u(x)v(y)$ for all $x, y \in S$.
\item $R'(x,y) = 0$ whenever $f(x) = f(y)$.
\end{itemize}
Then set $R(x,y)=\max\{R'(x,y),R'(y,x)\}$.

Let $X=\{x \mid u(x)>0\}$ and $Y=\{y \mid v(y)>0 \}$.
Note that for every $x\in S$, either $u(x)$ or $v(x)$ must be 0, as $R(x,x)$ must be 0, therefore $X \cap Y = \varnothing$. Then denote
\begin{equation*}
\CA_1(f) = \max_{u,v} \min_{x \in X, y \in Y, i \in [n] : \atop u(x)v(y) > 0, x_i \neq y_i} \max\{\theta(x, i), \theta(y, i)\}.
\end{equation*}
where $\theta(x,i)$ can be simplified to
\begin{equation*}
\theta(x, i) = \frac{\sum_{y \in Y} v(y)}{\sum_{y \in Y : x_i \neq y_i} v(y)} \hspace{1cm} \text{and} \hspace{1cm} \theta(y, i) = \frac{\sum_{x \in X} u(x)}{\sum_{x \in X : x_i \neq y_i} u(x)}.
\end{equation*}
Naturally, $\CA_1(f) \leq \CA(f)$.

As  $R(x,y)=0$ whenever $f(x)=f(y)$, we have that for every output $h \in H$ either $f^{-1}(h)\cap X = \varnothing$ or $f^{-1}(h)\cap Y = \varnothing$.
Therefore, $\CA_1(f)$ effectively bounds the complexity of differentiating between two non-overlapping sets of outputs.
This leads to the following equivalent definition for  $\CA_1(f)$:
\begin{proposition} \label{caa-prop}
Let $A \cup B = H$ be a partition of the output alphabet, i.e., $A \cap B = \varnothing$.
Let $p$ and $q$ be probability distributions over $X:=f^{-1}(A)$ and $Y:=f^{-1}(B)$, respectively.
Then
\begin{equation*}
\CA_1(f) = \max_{A, B \atop p, q} \min_{\substack{i \in [n], \\ g_1, g_2 \in G : g_1 \neq g_2 \\ \exists x \in X, y \in Y: p(x)q(y) > 0}} \frac{1}{\min\left\{\Pr_{x \sim p}[x_i \neq g_1], \Pr_{y \sim q}[y_i \neq g_2]\right\}}.
\end{equation*}
\end{proposition}

For the proof of this proposition see Appendix \ref{caa-def}.

\paragraph{Weighted adversary bound \cite{Amb03, LM04}.}
Let $w, w'$ be weight schemes as follows.
\begin{itemize}
\item Every pair $(x, y) \in S^2$ is assigned a non-negative weight $w(x, y) = w(y, x)$ such that $w(x, y) = 0$ whenever $f(x) = f(y)$.
\item Every triple $(x, y, i)$ is assigned a non-negative weight $w'(x, y, i)$ such that $w'(x, y, i) = 0$ whenever $x_i = y_i$ or $f(x) = f(y)$, and $w'(x, y, i), w'(y, x, i) \geq w(x, y)$ for all $x, y, i$ such that $x_i \neq y_i$.
\end{itemize}
For all $x, i$, let $wt(x) = \sum_{y \in S} w(x, y)$ and $v(x, i) = \sum_{y \in S} w'(x, y, i)$.
Denote
\begin{equation*}
\WA(f) = \max_{w,w'} \min_{x, y \in S, i \in [n] \atop w(x,y) \neq 0, x_i \neq y_i} \max\left\{ \frac{wt(x)}{v(x,i)}, \frac{wt(y)}{v(y,i)}\right\}.
\end{equation*}

\paragraph{Kolmogorov complexity \cite{LM04}.}
Let $\sigma \in \{0, 1\}^*$ be any finite string.\footnote{By the argument of \cite{SS06}, we take the minimum over the strings instead of the algorithms computing $f$.}
Denote
\begin{equation*}
\KA(f) = \min_\sigma \max_{x, y \in S \atop f(x) \neq f(y)} \frac{1}{\sum_{i : x_i \neq y_i} \min\{2^{-K(i|x, \sigma)},2^{-K(i|y, \sigma)}\}}.
\end{equation*}

\paragraph{Minimax over probability distributions \cite{LM04}.}
Let $\{p_x\}_{x \in S}$ be a set of probability distributions over $[n]$.
Denote
\begin{equation*}
\MM(f) = \min_p \max_{x, y \in S \atop f(x) \neq f(y)} \frac{1}{\sum_{i : x_i \neq y_i} \min\{p_x(i), p_y(i)\}}.
\end{equation*}

\section{Equivalence of the Adversary Bounds}
In this section we prove the main theorem: 

\begin{theorem} \label{thm:main}
Let $f : S \to H$ be a partial Boolean function, where $S \subseteq G^n$.
Then
\begin{itemize}
\item $\fbs(f) \leq \CA_1(f) \leq \CA(f) = \WA(f),$
\item $\WA(f) = O(\KA(f))$,
\item $\KA(f) = \Theta(\MM(f)).$
\end{itemize}
Moreover, for total functions $f : G^n \to H$, we have
\begin{equation*}
\fbs(f) = \MM(f).
\end{equation*}
\end{theorem}

The part $\WA(f) = O(\KA(f))$ has been already proven in \cite{LM04}.

\subsection{Fractional Block Sensitivity and the Weighted Adversary Method}

First, we prove that fractional block sensitivity lower bounds the relational adversary bound for any partial function.

\begin{proposition}
Let $f : S \to H$ be a partial Boolean function, where $S \subseteq G^n$.
Then
\begin{equation*}
\fbs(f) \leq \CA_1(f).
\end{equation*}
\end{proposition}

\begin{proof}
Let $x \in S$ be such that $\fbs(f, x) = \fbs(f)$ and denote $h = f(x)$.
Let $H' = H \setminus \{h\}$ and $S' = f^{-1}(H')$.

Let $\clB$ be the set of sensitive blocks of $x$.
Let $w : \clB \to [0, 1]$ be an optimal solution to the $\fbs(f,x)$ linear program, that is, $\sum_{B \in \clB} w(B) = \fbs(f, x)$.
For each $B \in \clB$, pick a single $y_B \in S'$ such that $B = \{i \mid x_i \neq y_i\}$.
Then define $R(x, y_B) := w(B)$ for all $B \in \clB$.
It is clear that $R$ has a corresponding rank 1 matrix $R'$, as it has only one row (corresponding to $x$) that is not all zeros.

Let $y \in S'$ be any input such that $R(x, y) > 0$.
Then for any $i \in [n]$ such that $x_i \neq y_i$,
\begin{equation*}
\theta(x, i) = \frac{\sum_{B \in \clB} w(B)}{\sum_{B \in \clB : i \in B} w(B)} = \frac{\fbs(f, x)}{\sum_{B \in \clB : i \in B} w(B)} \geq \fbs(f),
\end{equation*}
as $0 < \sum_{B \in \clB : i \in B} w(B) \leq 1$.
On the other hand, note that
\begin{equation*}
\theta(y, i) = \frac{w(B)}{w(B)} = 1,
\end{equation*}
where $B = \{i \mid x_i \neq y_i\}$.
Therefore, for this $R$,
\begin{equation*}
\min_{x, y \in S, i \in [n] : \atop R(x, y) > 0, x_i \neq y_i} \max\{\theta(x, i), \theta(y, i)\} \geq \min_{y \in S', i \in [n] : \atop R(x, y) > 0, x_i \neq y_i} \max\{\fbs(f), 1\} = \fbs(f),
\end{equation*}
and the claim follows.
\end{proof}

As mentioned in \cite{LM04}, $\CA(f)$ is a weaker version of $\WA(f)$.
We show that in fact they are exactly equal to each other:

\begin{proposition}
Let $f : S \to H$ be a partial Boolean function, where $S \subseteq G^n$.
Then
\begin{equation*}
\CA(f) = \WA(f).
\end{equation*}
\end{proposition}

\begin{proof}
\begin{itemize}
\item First we show that $\CA(f) \leq \WA(f)$.

Suppose that $R$ is the function for which the relational bound achieves maximum value.
Let $w(x,y) = w(y,x) = w(x,y,i) = w(y,x,i) = R(x,y)$ for any $x, y, i$ such that $f(x) \neq f(y)$ and $x_i \neq y_i$.
This pair of weight schemes satisfies the conditions of the weighted adversary bound.
The value of the latter with $w, w'$ is equal to $\CA(f)$.
As the weighted adversary bound is a maximization measure, $\CA(f) \leq \WA(f)$.

\item Now we show that $\CA(f) \geq \WA(f)$.

Let $w, w'$ be optimal weight schemes for the weighted adversary bound.
Let $R(x,y) = w(x,y)$ for any $x, y \in S$ such that $f(x) \neq f(y)$.
Let $S' = f^{-1}(H \setminus f(x))$.
Then
\begin{equation*}
\theta(x,i) = \frac{\sum_{y \in S'} R(x, y)}{\sum_{y \in S' : x_i \neq y_i} R(x, y)} = \frac{\sum_{y \in S'} w(x,y)}{\sum_{y \in S' : x_i \neq y_i} w(x,y)} \geq \frac{\sum_{y \in S'} w(x,y)}{\sum_{y \in S' : x_i \neq y_i} w'(x,y,i)} = \frac{wt(x)}{v(x,i)},
\end{equation*}
as $w'(x,y,i) \geq w(x,y)$ by the properties of $w, w'$.
Similarly, $\theta(y,i) \geq \frac{wt(y)}{v(y,i)}$.
Therefore, for any $x, y \in S$ and $i \in [n]$ such that $f(x) \neq f(y)$ and $x_i \neq y_i$, we have
\begin{equation*}
\max\{\theta(x,i),\theta(y,i)\} \geq \max\left\{ \frac{wt(x)}{v(x,i)}, \frac{wt(y)}{v(y,i)}\right\}.
\end{equation*}
As the relational adversary bound is also a maximization measure, $\CA(f) \geq \WA(f)$. \qedhere
\end{itemize}
\end{proof}

The proof of this proposition also shows why $\CA(f)$ and $\WA(f)$ are equivalent --- the weight function $w'$ is redundant in the classical case (in contrast to the quantum setting).

\subsection{Kolmogorov Complexity and Minimax over Distributions}

In this section we prove the equivalence between the mimimax over probability distributions and Kolmogorov complexity adversary bound.
It has been shown in the proof of the main theorem of  \cite{LM04} that $\MM(f) = \Omega(\KA(f))$.
Here we show the other direction using a well-known result from coding theory.

\begin{proposition}[Kraft's inequality]
Let $S$ be any prefix-free set of finite strings.
Then \begin{equation*}\sum_{x \in S} 2^{-|x|} \leq 1.\end{equation*}
\end{proposition}

\begin{proposition}
Let $f : S \to H$ be a partial Boolean function, where $S \subseteq G^n$.
Then
\begin{equation*}
\KA(f) \geq \MM(f).
\end{equation*}
\end{proposition}

\begin{proof}
Let $\sigma$ be the binary string for which $\KA(f)$ achieves the smallest value.
Define the set of probability distributions $\{p_x\}_{x \in S}$ on $[n]$ as follows.
Let $s_x = \sum_{i \in [n]} 2^{-K(i \mid x, \sigma)}$ and $p_x(i) = 2^{-K(i \mid x, \sigma)}/s_x$.
The set of programs that print out $i \in [n]$, given $x$ and $\sigma$, is prefix-free (by the definition of $\mathcal S$), as the information given to all programs is the same.
Thus by Kraft's inequality, we have $s_x \leq 1$.

Examine the value of the minimax bound with this set of probability distributions.
For any $x, y \in S$ and $i \in [n]$, we have
\begin{equation*}
\min\{p_x(i), p_y(i)\} = \min\left\{ \frac{2^{-K(i \mid x, \sigma)}}{s_x}, \frac{2^{-K(i \mid y, \sigma)}}{s_y}\right\} \geq \min\{2^{-K(i|x, \sigma)},2^{-K(i|y, \sigma)}\}.
\end{equation*}
Therefore, $\KA(f) = \Theta(\MM(f))$.
\end{proof}

\subsection{Fractional Block Sensitivity and Minimax over Distributions}

Now we proceed to prove that for total functions, fractional block sensitivity is equal to the minimax over probability distributions.
The latter has an equivalent form of the following program.
\begin{lemma} \label{cmmAlt} For any partial Boolean function $f : S \to H$, where $S \subseteq G^n$,
\begin{equation*}
\MM(f) = \min_v \max_{x \in S} \sum_{i \in [n]} v_x(i) \hspace{1cm} \text{ s.t. } \quad \forall y \in S \text{ s.t. } f(x) \neq f(y): \sum_{i : x_i \neq y_i} \min\{v_x(i), v_y(i)\} \geq 1,
\end{equation*}
where $\{v_x\}_{x \in S}$ is any set of weight functions $v_x : [n] \to \mathbb R_{\geq 0}$.
\end{lemma}

\begin{proof} Denote by $\mu$ the optimal value of the given program.
\begin{itemize}
\item First we prove that $\mu \leq \MM(f)$.

Construct a set of weight functions $\{v_x\}_{x \in S}$ by $v_x(i) := p_x(i) \cdot \MM(f)$, where $\{p_x\}_{x \in S}$ is an optimal set of probability distributions for the minimax bound.
Then for any $x, y$ such that $f(x) \neq f(y)$,
\begin{equation*}
\sum_{i : x_i \neq y_i} \min\{v_x(i), v_y(i)\} = \MM(f) \cdot \sum_{i : x_i \neq y_i} \min\{p_x(i), p_y(i)\} \geq \MM(f) \cdot \frac{1}{\MM(f)} = 1.
\end{equation*}
On the other hand, the value of this solution is given by
\begin{equation*}
\max_{x \in S} \sum_{i \in [n]} v_x(i) = \max_{x \in S} \MM(f) \cdot \sum_{i \in [n]} p_x(i) = \MM(f).
\end{equation*}

\item Now we prove that $\mu \geq \MM(f)$.

Let $\{v_x\}_{x\in S}$ be an optimal solution for the given program.
Set $s_x = \sum_{i \in [n]} v_x(i)$.
Construct a set of probability distributions $\{p_x\}_{x\in S}$ by $p_x(i) = v_x(i)/s_x$.
Then for any $x, y$ such that $f(x) \neq f(y)$, we have
\begin{equation*}
\sum_{i : x_i \neq y_i} \min\{p_x(i), p_y(i)\} = \sum_{i : x_i \neq y_i} \min\left\{\frac{v_x(i)}{s_x}, \frac{v_y(i)}{s_y}\right\} \geq \frac{1}{\mu} \cdot \sum_{i : x_i \neq y_i} \min\left\{v_x(i), v_y(i)\right\} \geq \frac{1}{\mu}.
\end{equation*}
Therefore, $\MM(f) \leq \mu$. \qedhere
\end{itemize}
\end{proof}

In this case we prove that for total functions the minimax over probability distributions is equal to the fractional certificate complexity $\FC(f)$.
The result follows since $\FC(f) = \fbs(f)$.
The proof of this claim is almost immediate in light of the following ``fractional certificate intersection'' lemma by Kulkarni and Tal:
\begin{proposition}[\cite{KT16}, Lemma 6.2] \label{KulkarniTal}
Let $f : G^n \to H$ be a total function\footnote{Kulkarni and Tal prove the lemma for Boolean functions, but it is straightforward to check that their proof also works for functions with arbitrary input and output alphabets.} and $\{v_x\}_{x \in G^n}$ be a feasible solution for the $\FC(f)$ linear program.
Then for any two inputs $x, y \in G^n$ such that $f(x) \neq f(y)$, we have
\begin{equation*}
\sum_{i : x_i \neq y_i} \min\{v_x(i), v_y(i)\} \geq 1.
\end{equation*}
\end{proposition}

Let $f$ be a total function.
Suppose that $\{v_x\}_{x \in G^n}$ is a feasible solution for the $\MM(f)$ program.
Then for any $x, y \in G^n$ such that $f(x) \neq f(y)$,
\begin{equation*}
\sum_{i : x_i \neq y_i} v_x(i) \geq \sum_{i : x_i \neq y_i} \min\{v_x(i), v_y(i)\} \geq 1.
\end{equation*}
Hence this is also a feasible solution for the $\FC(f)$ linear program.
On the other hand, if $\{v_x\}_{x \in G^n}$ is a feasible solution for $\FC(f)$ linear program, then it is also a feasible solution for the $\MM(f)$ program by Proposition \ref{KulkarniTal}.
Therefore, $\MM(f) = \FC(f)$.

\section{Separations for Partial Functions}

\subsection{Fractional Block Sensitivity vs. Adversary Bounds}

Here we show an example of a partial function that provides an unbounded separation between the adversary measures and fractional block sensitivity.

\begin{theorem}
There exists a partial Boolean function $f : S \to \{0, 1\}$, where $S \subseteq \{0, 1\}^n$, such that
$\fbs(f) = O(1)$ and $\CA_1(f), \CA(f), \WA(f), \KA(f), \MM(f) = \Omega(n)$.
\end{theorem}

\begin{proof}
Let $n$ be an even number and $S = \{x \in \{0, 1\}^n \mid |x| = 1\}$ be the set of bit strings of Hamming weight 1.
Define the ``greater than half'' function $\GTH_n : S \to \{0, 1\}$ to be 1 iff $x_i = 1$ for $i > n/2$.

For the first part, the certificate complexity is constant $\C(\GTH_n) = 1$.
To certify the value of greater than half, it is enough to certify the position of the unique $i$ such that $x_i = 1$.
The claim follows, as $\C(f) \geq \fbs(f)$ for any $f$.

For the second part, by Theorem \ref{thm:main}, it suffices to show that $\CA_1(\GTH_n) = \Omega(n)$.
Let $X = f^{-1}(0)$ and $Y = f^{-1}(1)$.
Let $R(x,y) = 1$ for all $x \in X, y \in Y$.
Suppose that $x \in X, y \in Y, i \in [n]$ are such that $x_i = 1$ (and thus $y_i = 0$).
Then
\begin{align*}
\theta(x, i) &= \frac{\sum_{y^* \in Y} R(x, y^*)}{\sum_{y^* \in Y : x_i \neq y^*_i} R(x, y^*)} = \frac{n/2}{n/2} = 1,\\
\theta(y, i) &= \frac{\sum_{x^* \in X} R(x^*, y)}{\sum_{x^* \in X : x^*_i \neq y_i} R(x^*, y)} = \frac{n/2}{1} = n/2.
\end{align*}
Therefore, $\max\{\theta(x,i), \theta(y,i)\} = n/2$.
Similarly, if $i$ is such an index that $y_i = 1$ and $x_i = 0$, we also have $\max\{\theta(x,i), \theta(y,i)\} = n/2$.
Also note that $R$ has a corresponding rank 1 matrix $R'$, hence $\CA_1(f) \geq n/2 = \Omega(n)$.
\end{proof}

We note that a similar function was used to prove lower bounds on the problem of inverting a permutation \cite{Amb00, Aar06}.
More specifically, we are given a permutation $\sigma(1), \ldots, \sigma(n)$, and the function is 0 if $\sigma^{-1}(1) \leq n/2$ and 1 otherwise.
With a single query, one can find the value of $\sigma(i)$ for any $i$.
By construction, a lower bound on $\GTH_n$ also gives a lower bound on computing this function.

\subsection{Relational Adversary vs. Kolmogorov Complexity Bound}

Here we show that, for a variant of the ordered search problem, the Kolmogorov complexity bound gives a tight logarithmic lower bound, while the relational adversary gives only a constant value lower bound.

\begin{theorem}
There exists a partial Boolean function $f : S \to \{0, 1\}$, where $S \subseteq \{0, 1\}^n$, such that
$\CA_1(f), \CA(f), \WA(f) = O(1)$ and $\KA(f), \MM(f) = \Omega(\log n)$.
\end{theorem}

\begin{proof}
Let $S = \{x \in \{0, 1\}^n \mid \exists i \in [0;n]: x_1 = \ldots x_i = 0 \text{ and } x_{i+1} = \ldots = x_n = 1\}$.
In other words, $x$ is any string starting with some number of 0s followed by all 1s.
Define the ``ordered search parity'' function $\OSP_n : S \to \{0, 1\}$ to be $\IND(x) \bmod 2$, where $\IND(x)$ is the last index $i$ such that $x_i = 0$ (in the special case $x = 1^n$, assume that $i = 0$).

For simplicity, further assume that $n$ is even.
First, we prove that $\KA(f) = \Omega(\log n)$.
We use the argument of Laplante and Magniez and the distance scheme method they have adapted from \cite{HNS01}:
\begin{proposition}[\cite{LM04}, Theorem 5] \label{distance-scheme}
Let $f : S \to \{0, 1\}$ be a Boolean function, where $S \subseteq \{0, 1\}^n$.
Let $D$ be a non-negative integer function on $S^2$ such that $D(x, y) = 0$ whenever
$f(x) = f(y)$. Let $W = \sum_{x,y:D(x,y)\neq 0} \frac{1}{D(x,y)}$.
Define the \emph{right load} $\RL(x,i)$ to be the maximum over all values $d$, of the number of $y$ such that $D(x,y) = d$ and $x_i \neq y_i$.
The \emph{left load} $\LL(y,i)$ is defined similarly, inverting $x$ and $y$. Then
\begin{equation*}
\KA(f) = \Omega\left( \frac{W}{|S|} \min_{x,y,i \atop D(x,y)\neq 0, x_i \neq y_i} \max\left\{\frac{1}{\RL(x,i)}, \frac{1}{\LL(y,i)}\right\}\right).
\end{equation*}
\end{proposition}

For each pair $x, y$ such that $f(x) \neq f(y)$ and $\IND(x) > \IND(y)$, let $D(x,y) = \IND(x)-\IND(y)$.
Then we have
\begin{equation*}
W = \sum_{k = 1}^{n/2} ((n+1)- (2k-1)) \frac{1}{2k-1} = (n+1) \sum_{k=1}^{n/2} \frac{1}{2k-1} - \frac n 2.
\end{equation*}
Since $\sum_{k=1}^{n/2} 1/(2k-1) > \sum_{k=1}^{n/2} 1/2k = \frac 1 2 \cdot \sum_{k = 1}^{n/2} 1/k = H_{n/2} = \Theta(\log n)$ as a harmonic number, we have that $W > (n+1)H_{n/2} - n/2= \Theta(n \log n)$.

On the other hand, since for every $x \in S$ and positive integer $d$ there is at most one $y$ such that $D(x,y) = d$, we have that $\RL(x,i) = \LL(y,i) = 1$ for any $x, y$ such that $f(x) \neq f(y)$ and $x_i \neq y_i$.
Since $|S| = n+1$, by Proposition \ref{distance-scheme},
\begin{equation*}
\KA(\OSP_n) = \Omega\left( \frac{n \log n}{n} \right) = \Omega(\log n).
\end{equation*}

Now we prove that $\CA(\OSP_n)  \leq 2$.
Let $ N=n/2 $; we start by fixing an enumeration of $ S $.
By $ x^{(i)}  $, $ i \in [N+1] $, we denote the unique element of $ S $ satisfying $ \IND(x^{(i)})  = 2i-2$ (it is a negative input for $ \OSP_n $); 
by $ y^{(j)}  $, $ j \in [N] $, we denote the unique element of $ S $ satisfying $ \IND(y^{(j)})  = 2j-1$ (it is a positive input for $ \OSP_n $).

We claim that for every $ R = (r_{ij})$,  ${i \in [N+1], j\in [N]} $, with nonnegative entries we have
\[ 
\min_{ \substack{(i,j) \in [N+1] \times [N] : \\ r_{ij}>0 }  }   \min_{\substack{t \in [n] : \\ x^{(i)}_t \neq  y^{(j)} _t}} \max  \{\theta(   x^{(i)} ,t),   \theta(y^{(j)},t)  \} \leq 2,
\]
unless $ r_{ij} =0$ for all $ i,j $.  Since $ \CA(\OSP_n) $  is defined only for  $ R $ which are not identically zero, we conclude that $ \CA(f) \leq 2 $.

For all $ i \in [N+1] $, $ j=[N]$ we set 
\[ 
t_{ij} = \min \{t  :      x^{(i)}_{t} \neq  y^{(j)} _{t}  \}
= 
1+ \min \{ \IND( x^{(i)}),  \IND( y^{(j)})    \}
=
\begin{cases}
2i-1, &  i\leq j, \\
2j, &  i>j.
\end{cases}.
\]
We shall show that, unless $ R \equiv 0 $, there is a pair $ (i,j)  $ satisfying
\begin{equation}\label{eq:p11e01}
r_{ij} > 0 
\quad\text{and}\quad
 \max  \{\theta(   x^{(i)} , t_{ij}),   \theta(y^{(j)},t_{ij})  \} \leq 2.
\end{equation}
Consider $ i \in  \{2,3,\ldots,N+1\}  $ and $ j \in [i-1] $.  
Then we have $ t_{ij}  = 2j $    and
\begin{equation}\label{eq:p11e02}
\theta(   x^{(i)} , t_{ij})= \frac{\sum_{k=1}^N  r_{ik}}{\sum_{k=1}^{j}  r_{ik}  }, 
\quad
 \theta(y^{(j)},t_{ij}) = \frac{\sum_{l=1}^{N+1}  r_{lj} }{\sum_{l=j+1}^{N+1}  r_{lj}  }.
\end{equation}
Now consider  $ i \in  [N] $ and $ j \in \{i,i+1,\ldots,N\} $.  
Then we have $ t_{ij}  = 2i-1 $    and
\begin{equation}\label{eq:p11e03}
\theta(   x^{(i)} , t_{ij})= \frac{\sum_{k=1}^N  r_{ik}}{\sum_{k=i}^{N}  r_{ik}  },
\quad
\theta(y^{(j)},t_{ij}) =  \frac{\sum_{l=1}^{N+1}  r_{lj} }{\sum_{l=1}^{i}  r_{lj}  }.
\end{equation}

We introduce the following notation:
\begin{itemize}
	\item $ \alpha_{ij} = \sum_{k = j+1}^N r_{ik} $ 	and 	$ \beta_{ij} = \sum_{k=1}^j r_{ik}$, for $ i \in [N+1]$  and $ j \in \{0,1,\ldots,i-1\} $;
	\item $ \gamma_{ij} = \sum_{l=i+1}^{N+1}  r_{lj} $ and  $ \delta_{ij} = \sum_{l=1}^{i} r_{lj} $  for   $ i \in [N] $,  $ j \in \{i,i+1,\ldots,N\} $.
\end{itemize}
 By convention, $ \beta_{10} = \alpha_{N+1,N} = 0$. 
 Then \eqref{eq:p11e02}--\eqref{eq:p11e03} can be rewritten as follows:
\[ 
\theta(   x^{(i)} , t_{ij}) = 
\begin{cases}
1 +   \alpha_{ij}  / \beta_{ij}, & j < i , \\
1 +    \beta_{i,i-1} / \alpha_{i, i-1}, & j \geq  i ,
\end{cases}
\qquad
\theta(y^{(j)},t_{ij}) = 
\begin{cases}
1 +   \delta_{jj}  / \gamma_{jj}, & j < i , \\
1 +    \gamma_{ij} / \delta_{ij}, & j \geq  i .
\end{cases}
 \]
 
Consequently,   \eqref{eq:p11e01} holds if    there is a pair $ (i,j) \in [N+1] \times [N] $ such that $ r_{ij} >0$ and
\[ 
\begin{cases}
  \left(  \alpha_{ij}  \leq  \beta_{ij}  \right) \land \left(  \delta_{jj}   \leq  \gamma_{jj} \right), & j <i,\\
  \left(\beta_{i,i-1}  \leq   \alpha_{i, i-1}  \right)  \land  \left(   \gamma_{ij}   \leq  \delta_{ij}  \right) , &  j \geq i.
\end{cases}
 \]
 Suppose the contrary: for all  $ (i,j) \in [N+1] \times [N] $  we have
 \begin{equation}\label{c:C1}\tag{C1}
 i> j \Rightarrow   
\left(   r_{ij} = 0  \right)
\lor 
\left( \alpha_{ij} > \beta_{ij}   \right) 
\lor 
\left( \delta_{jj}  > \gamma_{jj}   \right)
 \end{equation}
 and
 \begin{equation}\label{c:C2}\tag{C2}
 i \leq  j \Rightarrow   
 \left( r_{ij} = 0\right) 
 \lor 
 \left(\beta_{i,i-1} > \alpha_{i, i-1}\right)  
 \lor 
\left(\gamma_{ij} >\delta_{ij}  \right)
 .
 \end{equation}
 We shall show by induction   that     for all $ i \in \{0,1,\ldots,N\}$, $ j\in [N] $ the following holds:
 \begin{equation}\label{eq:p11e04}
 \alpha_{i +1,i }  \geq  \beta_{i+1,i  }
 \quad\text{and}\quad
 \gamma_{jj} \geq \delta_{jj}. 
 \end{equation}
 When that is established, it follows that all $ r_{ij} $ must be zero. To see that, recall $ \alpha_{N+1,N }  =0 $. Since \eqref{eq:p11e04}  implies $   \beta_{N+1,N } \leq  \alpha_{N+1,N }  =0$, we obtain
$  \beta_{N+1,N }=\sum_{k=1}^N r_{N+1,k} \leq 0 $. However, all $ r_{lk} $ are nonnegative, hence $ r_{N+1,k}=0   $ for all $ k \in [N]$. That, in turn, implies	
$ \sum_{l=1}^{N} r_{lN} = \delta_{NN}
\leq 
\gamma_{NN} =  r_{N+1,N} = 0 $, where we    have used  \eqref{eq:p11e04} again. Now $ r_{lN} = 0 $ for all $ l \in [N]$ (and also for $ l=N+1 $), thus $ \alpha_{N,N -1} =  r_{NN} = 0 $.
Continue inductively to obtain that 
$ \alpha_{i+1,i }  =  \beta_{i+1,i }=0$ and $\gamma_{jj} = \delta_{jj} =0  $ (and $ r_{ij} = 0 $)   for all $ i,j $.
It remains to show \eqref{eq:p11e04}.

\textbf{The base case:} we already have $  \alpha_{10} \geq 0 = \beta_{10} $. 
For \textbf{the inductive step}, suppose that \eqref{eq:p11e04} holds for all $ i   \in \{0,1,\ldots,p-1\}  $ and $ j \in [p-1] $, for some $ p \in [N]$ (for $ p=1 $, the inequality  $  \gamma_{jj} \geq \delta_{jj} $ remains  unproven for all $ j $). We shall show that both inequalities hold also with $ i=j=p $. The proof is by  contradiction.

Suppose that $ \gamma_{pp} < \delta_{pp} $. From \eqref{c:C2} it follows that either $ r_{pp} = 0 $ or $ \beta_{p,p-1} > \alpha_{p,p-1} $. The latter is false by the inductive hypothesis, thus  $ r_{pp} = 0 $.
But then
\[ 
\gamma_{p-1,p} =   \sum_{l=p}^{N+1}  r_{lp} = \gamma_{pp} 
\quad\text{and}\quad
\delta_{p-1,p} = \sum_{l=1}^{p-1} r_{lp}= \delta _{pp} .
\]
Thus we have $ \gamma_{p-1,p} < \delta_{p-1,p} $. Again, from \eqref{c:C2} it follows that either $ r_{p-1,p} = 0 $ or $ \beta_{p-1,p-2} > \alpha_{p-1,p-2} $. The latter is false, thus   $ r_{p-1,p} = 0 $, which implies
$ \gamma_{p-2,p}= \gamma_{p-1,p} < \delta_{p-1,p} = \delta_{p-2,p} $.
Continuing similarly, we obtain
$ r_{1p}=r_{2p} = \ldots = r_{pp} =0 $.
However, then $ \delta_{pp} =0$ and the inequality  $ \gamma_{pp} < \delta_{pp} $ is impossible, a contradiction. 

Suppose that $ \beta_{p+1,p} >  \alpha_{p+1,p} $. From \eqref{c:C1} it follows  that either $ r_{p+1,p} = 0 $ or $ \delta_{pp} > \gamma_{pp} $. As shown previously, the latter is false, thus $ r_{p+1,p}=0 $. But then we have
\[ 
\alpha_{p+1,p-1} = \sum_{k = p}^N r_{p+1,k} =  \alpha_{p+1,p} 
\quad\text{and}\quad
\beta_{p+1,p-1}  = \sum_{k=1}^{p-1} r_{p+1,k} = \beta_{p+1,p}.
\]
Hence we also have  $ \beta_{p+1,p-1} >  \alpha_{p+1,p-1} $. Then again from \eqref{c:C1} we either have  $ \delta_{p-1,p-1} > \gamma_{p-1,p-1} $, or $ r_{p+1,p-1} = 0 $. The former is false by the inductive hypothesis, the latter implies $ \beta_{p+1,p-2}  = \beta_{p+1,p-1} >  \alpha_{p+1,p-1} =  \alpha_{p+1,p-2}  $.
Continuing similarly, we obtain
$ r_{p+1,1}  = \ldots = r_{p+1,p} =0 $.
But then  $ \beta_{p+1,p} = 0\leq   \alpha_{p+1,p} $, a contradiction.
This completes the inductive step. \qedhere
\end{proof}

\subsection{Rank-1 Adversary vs. Relational Adversary}

In this section we show a function such that the relational adversary bound $\CA(f)$ is quadratically larger than the rank-1 relational adversary $\CA_1(f)$.
First we give an example of a non-Boolean function, and then convert it to a Boolean function with the same separation.

\begin{theorem} \label{thm:cra1cra}
There exists a function $f : S \to \mathbb N$, where $S \subseteq \{0,1\}^n$, such that $\CA(f) = \Omega(n)$ and $\CA_1(f) = O(\sqrt n)$.
\end{theorem}

\begin{proof}
Let $n$ be a perfect square and $N^2 = n$.
For an input $x \in \{0, 1\}^n$, split it into $N$ blocks of $N$ consecutive bits, and denote the $j$-th bit in the $i$-th block by $x_{ij}$.
Then define $S$ to be the set of all inputs $x$ such that the Hamming weight of each block is exactly 1.
Let $f$ be any injection on $S$.

First, we prove that $\CA(f) = \Omega(N^2)$.
Let $R(x,y) = 1$ iff $x$ and $y$ differ in exactly 2 bits.
Pick any two such inputs $x$ and $y$, and a position $i$ such that $x_i \neq y_i$.
W.l.o.g.~assume that $x_i = 0$.
Examine $\theta(x,i) = \frac{\sum_{z \in S}R(x,z)}{\sum_{z \in S,z_i \neq x_i}R(x,z)}$.
\begin{itemize}
\item The number of $z$ such that $x$ and $z$ differ in 2 bits is $N(N-1)$, since we can pick any of the $N$ blocks of $x$ and change the position of the single 1 in that block to any of $N-1$ other positions.
Hence, $\sum_{z \in S} R(x,z) = N(N-1)$.
\item There is only one $z$ such that $z_i \neq x_i$ and $x$ and $z$ differ in exactly two bits, as $z_i = 1$.
\end{itemize}
Thus, $\sum_{z \in S,z_i \neq x_i} R(x,z) = 1$ and $\theta(x,i) = N(N-1)/1$.
Therefore, for any $x,y,i$ such that $R(x,y) > 0$ and $x_i \neq y_i$, we have $\max(\theta(x,i),\theta(y,i)) = N(N-1)$, and $\CA(f) = \Omega(N^2)$.

Now we prove that $\CA_1(f) \leq N$.
By Proposition \ref{caa-prop}, let $X$, $Y$ be the partition of $S$ and $u : X \to \mathbb R$, $v : Y \to \mathbb R$ be the probability distributions that achieve $\CA_1(f)$ 
(e.g., $\sum_{x \in X} u(x) = \sum_{y \in Y} v(y) = 1$).
For $g : S \to \mathbb R$, $i \in [n]$, $b \in \{0, 1\}$, define
$$s(g,i,b) = \sum_{x \in g^{-1} \atop x_i = b} g(x).$$
Then $\theta(x,i) = \frac{1}{s(v,i,1-x_i)}$ and $\theta(y,i) = \frac{1}{s(u,i,1-y_i)}$.
We prove the following lemma:
\begin{lemma} \label{thm:cra1lemma}
For all $i \in [n]$, there is a value $b \in \{0, 1\}$ such that
$$s(u,i,b) \leq \frac{1}{\CA_1(f)} \text{\quad and\quad} s(v,i,b) \leq \frac{1}{\CA_1(f)}.$$
\end{lemma}
\begin{proof}
Let $p := 1/\CA_1(f)$.
Assume on the contrary that for each $b \in \{0, 1\}$, either $s(u,i,b) > p$ or $s(v,i,b) > p$.
We distinguish two cases:
\begin{itemize}
\item For some $b$, we have $s(v,i,b) > p$ and $s(u,i,1-b) > p$.
Then we can pick $x \in X$, $y \in Y$ such that $x_i = b$, $y_i = 1-b$ and $u(x)v(y) > 0$.
We have
$$\max\{\theta(x,i),\theta(y,i)\} = \max\left\{\frac{1}{s(v,i,b)}, \frac{1}{s(u,i,1-b)}\right\} < \frac 1 p = \CA_1(f),$$
a contradiction.
\item W.l.o.g., $s(u,i,0) > p$, $s(u,i,1) > p$, $s(v,i,0) \leq p$ and $s(v,i,1) \leq p$.
In that case
$$2p < s(u,i,0) + s(u,i,1) = 1 = s(v,i,0) + s(v,i,1) \leq 2p,$$
a contradiction.\qedhere
\end{itemize}
\end{proof}

Now assume on the contrary that $\CA_1(f) > N$.
For $b \in \{0, 1\}$, let $\overline b := 1-b$.
Suppose that $b_i$ is the value that satisfies the conditions of Lemma \ref{thm:cra1lemma} for $i \in [n]$.
Define $z := \overline{b_1} \overline{b_2} \ldots \overline{b_n}$.

First, we prove that $z \in S$.
Pick any $i \in [N]$ (any block).
Let $B = \{(i-1)N+1,\ldots,iN\}$ be the set of variables of the $i$-th block.
Then $$\sum_{j \in B} s(u,j,\overline{z_j}) = \sum_{j \in B} s(u,j,b_j) \leq N \cdot \frac{1}{\CA_1(f)} < 1$$ by the lemma and the assumption.
Since $\sum_{x \in X} u(x) = 1$, there is an $x \in X$ such that $x_{ij} = z_{ij}$ for all $j \in [N]$, thus the $i$-th block of $z$ is a correct Hamming weight 1 block.
Since we picked $i$ arbitrarily, each block of $z$ is correct and $z \in S$.

Now, we prove that $z \in X$.
Examine any $x \in X$ that is not $z$.
The inputs $x$ and $z$ differ in at least one block, hence they have 1s in different positions in that block.
Thus there is a position $i$ such that $z_i = 1$ and $x_i = 0$.
Therefore, we have
$$\sum_{x \in X \atop x \neq z} u(x) \leq \sum_{i : z_i = 1} s(u,i,\overline{z_i}) \leq N \cdot \frac{1}{\CA_1(f)} < 1$$
by the lemma and the assumption.
Since $\sum_{x \in X} u(x) = 1$, it follows that $u(z) > 0$, thus $z \in X$.
Similarly, we prove that $z \in Y$ and we get a contradiction.
\end{proof}

We can extend this result to Boolean functions:
\begin{theorem}
There exists a Boolean function $f : S \to \{0,1\}$, where $S \subseteq \{0,1\}^n$, such that $\CA(f) = \Omega(n)$ and $\CA_1(f) = O(\sqrt n)$.
\end{theorem}

\begin{proof}
Let $S$ be the same as in Theorem \ref{thm:cra1cra}.
Define $f$ as
$$f(x) = \left(\sum_{i \in [n]} i \cdot x_i\right) \bmod 2.$$

For $\CA(f)$, now define $R(x,y) = 1$ iff $f(x) \neq f(y)$ and $x$ and $y$ differ in exactly 2 bits.
For any $x$, we can change the position of any 1 in any block to a position of a different parity in that block in either $\lfloor N/2 \rfloor$ or $\lceil N/2 \rceil$ ways.
Therefore, $\sum_{y \in S} R(x,y) \geq N\cdot \lfloor N/2 \rfloor = \Omega(N^2)$.
By the same argument as in the previous proof, we have $\CA(f) = \Omega(1)$.

On the other hand, the argument for the rank-1 adversary from the previous proof works for any $X$, $Y$ (in this case, $X = f^{-1}(0)$, $Y = f^{-1}(1)$).
Hence, we still have $\CA_1(f) = O(N)$.
\end{proof}

\section{Limitation of Fractional Block Sensitivity}

In this section we show that there is a certain barrier that the fractional block sensitivity cannot overcome for partial functions.

\subsection{Upper Bound in Terms of Block Sensitivity}

\begin{theorem}
For any partial function $f : S \to H$, where $S \subseteq G^n$, and any $x \in S$,
\begin{equation*}
\fbs(f) \leq \sqrt{n \cdot \bs(f)}. 
\end{equation*}
\end{theorem}

\begin{proof}
We will prove that $\fbs(f,x) \leq \sqrt{n \cdot \bs(f,x)}$ for any $x \in S$.
First we introduce a parametrized version of the fractional block sensitivity.
Let $x \in S$ be any input, $\clB$ the set of sensitive blocks of $x$ and $N \leq n$ a positive real number.
Define
\begin{align*}
\fbs_N(f, x) = \max_w \sum_{B \in \clB} w(B) \hspace{1cm} \text{s.t.} \quad &\forall i \in [n]: \sum_{B \in \clB : i \in B} w(B) \leq 1,\\
													&\sum_{B \in \clB} |B|\cdot w(B) \leq N.
\end{align*}
where $w : \clB \to [0; 1]$.
If we let $N = n$, then the second condition becomes redundant and $\fbs_n(f,x) = \fbs(f,x)$.

For simplicity, let $k = \bs(f, x)$.
We will prove by induction on $k$ that $\fbs_N(f, x) \leq \sqrt{N k}$.
If $k= 0$, the claim obviously holds, so assume $k > 0$.
Let $\ell$ be the length of the shortest block in $\clB$.
Then
\begin{equation*}
\sum_{B \in \clB} \ell \cdot w(B) \leq \sum_{B \in \clB} |B|\cdot w(B) \leq N
\end{equation*}
and $\fbs_N(f, x) = \sum_{B \in \clB} w(B) \leq N/\ell$.

On the other hand, let $D$ be any shortest sensitive block.
Let $f'$ be the restriction of $f$ where the variables with indices in $D$ are fixed to the values of $x_i$ for all $i \in D$.
Note that $\bs(f', x) \leq k-1$, as we have removed all sensitive blocks that overlap with $D$.
Let $\clB'$ be the set of sensitive blocks of $x$ on $f'$ and let $\clT = \{B \in \clB \mid B \cap D \neq \varnothing\}$, the set of sensitive blocks that overlap with $D$ (including $D$ itself).
Then no $T \in \clT$ is a member of $\clB'$, therefore
\begin{equation*}
\sum_{B' \in \clB'} |B'|\cdot w(B') \leq N - \sum_{T \in \clT} |T| \cdot w(T) \leq N - \ell \cdot \sum_{T \in \clT} w(T).
\end{equation*}

Denote $t = \sum_{T \in \clT} w(T)$.
We have that $t \leq |D| = \ell$, as any $T \in \clT$ overlaps with $D$.
By combining the two inequalities we get
\begin{align*}
\fbs_N(f, x) &\leq \max_{\ell \in [0; n]} \min\left\{\frac N \ell, \max_{t \in [0; \ell]} \left\{ t + \fbs_{N-\ell t}(f', x)\right\} \right\} \\
& \leq \max_{\ell \in [0; n]} \min\left\{\frac N \ell, \max_{t \in [0; \ell]} \left\{ t + \sqrt{(N-\ell t)(k-1)}\right\} \right\}.
\end{align*}
If $N/\ell \leq \sqrt{N k}$, we are done.
Thus further assume that $\ell < \sqrt{N/k}$.

Denote $g(t) = t + \sqrt{(N-\ell t)(k-1)}$.
We need to find the maximum of this function on the interval $[0;\ell]$ for a given $\ell$.
Its derivative,
\begin{equation*}
g'(t) = 1 - \frac \ell 2 \sqrt{\frac{k-1}{N-\ell t}},
\end{equation*}
is a monotone function in $t$.
Thus it has exactly one root,
$
t_0 = N/\ell - (k-1) \cdot \ell/4.
$
Therefore, $g(t)$ attains its maximum value on $[0;\ell]$ at one of the points $\{0, t_0, \ell\}$.
\begin{itemize}
\item If $t = 0$, then $g(0) = \sqrt{N(k-1)} \leq \sqrt{Nk}$.
\item If $t = t_0$, then, as $t \leq \ell < \sqrt{N/k}$,
\begin{align*}
\sqrt{Nk} - \frac{k-1}{4} \cdot \sqrt{\frac{N}{k}} &< \frac N \ell -  (k-1) \frac \ell 4 < \sqrt{\frac N k} \\
\sqrt k - \frac{k-1}{4 \sqrt k} &< \sqrt{\frac 1 k} \\
3k &< 0.
\end{align*}
The last inequality has no solutions in natural numbers for $k$, so this case is not possible.
\item If $t = \ell$, then $g(t) = \ell + \sqrt{(N-\ell^2)(k-1)}$.
\end{itemize}
Now it remains to find the maximum value of $h(k) = \ell + \sqrt{(N-\ell^2)(k-1)}$ on the interval $[0;\sqrt{N/k}]$.
The derivative is equal to
\begin{equation*}
h'(\ell) = 1-\ell\cdot \sqrt{\frac{k-1}{N-\ell^2}}.
\end{equation*}
The only non-negative root of $h'(\ell)$ is equal to $\ell_0 = \sqrt{N/k}$.
Then $h(\ell)$ is monotone on the interval $[0; \sqrt{N/k}]$.
Thus $h(\ell)$ attains its maximal value at one of the points $\{0, \sqrt{N/k}\}$.
\begin{itemize}
\item If $\ell = 0$, then $h(\ell) = \sqrt{N(k-1)} < \sqrt{Nk}$.
\item If $\ell = \ell_0 = \sqrt{N/k}$, then
\begin{equation*}
h(\ell) = \sqrt{\frac N k} + \sqrt{\left(N-\frac N k\right)(k -1)} = \sqrt N \left(\sqrt{\frac 1 k} + (k-1) \sqrt{\frac 1 k}\right) = \sqrt{Nk}.
\end{equation*}
\end{itemize}
Thus, $h(\ell) \leq \sqrt{Nk}$ and that concludes the induction.

Therefore, $\fbs(f,x) = \fbs_n(f,x) \leq \sqrt{n \cdot \bs(f,x)}$, hence also $\fbs(f) \leq \sqrt{n \cdot \bs(f)}$ and we are done.
\end{proof}

We also give a simpler proof of the same (asymptotically) upper bound:

\begin{theorem}
For any partial function $f : S \to H$, where $S \subseteq G^n$, and any $x \in S$,
\begin{equation*}
\fbs(f) = O(\sqrt{n \cdot \bs(f)}). 
\end{equation*}
\end{theorem}

\begin{proof}
We show that for all $x \in S$, we have $\FC(f,x) = O(\sqrt{n \cdot \bs(f,x)})$.
The claim then follows as $\fbs(f,x) = \FC(f,x)$. 

Since $\FC(f,x)$ is a minimization linear program, it suffices to show a fractional certificate $v$ of size at most $O(\sqrt{n\cdot \bs(f,x)})$.
Let $k$ be a parameter between 1 and $n$.
Let $\clB = \{B \subseteq [n] \mid f(x) \neq f(x^B), |B| \leq k\}$ be a maximum set of non-overlapping sensitive blocks of $x$ of size at most $k$.
Then $|\clB| \leq \bs(f)$.
Let $S = \bigcup_{B \in \clB} B$ be the set of all positions in blocks of $\clB$.
We construct the fractional certificate $v$ by setting $v(i) = 1$ for all $i \in S$, and $v(i) = 1/k$ for all $i \notin S$.

Let $B$ be any sensitive block of $x$ of size at most $k$.
As $\clB$ is a maximum set of non-overlapping sensitive blocks, there must exist a $B' \in \clB$ such that $B \cap B' \neq \varnothing$.
Therefore, $\sum_{i \in B} v(i) \geq |B \cap B'| \geq 1$.
On the other hand, if $|B| \geq k$, then $\sum_{i \in B} v(i) \geq |B|/k \geq 1$.
Hence $v$ is a feasible fractional certificate.
The size of $v$ is $\sum_{i \in [n]} v(i) \leq |\clB|\cdot k + n/k \leq \bs(f) \cdot k + n/k$.
The last expression asymptotically reaches the minimum when $\bs(f) \cdot k = n/k$, which happens if $k = \sqrt{n / \bs(f)}$.
Then $\FC(f,x) = O(\sqrt{n\cdot \bs(f)})$.
\end{proof}

\subsection{A Matching Construction}

\begin{theorem}
For any $k \in \mathbb N$, there exists a partial Boolean function $f : S \to \{0, 1\}$, where $S \subseteq \{0, 1\}^n$, such that $\bs(f) = k$ and $\fbs(f) = \Omega(\sqrt{n \cdot \bs(f)})$.
\end{theorem}

\begin{proof}
Take any finite projective plane of order $t$, then it has $\ell = t^2+t+1$ many points.
Let $n = k\ell$ and enumerate the points with integers from 1 to $\ell$.
Let $X = \{0^\ell\}$ and $Y = \{y \mid$ there exists a line $L$ such that $y_i = 1$ iff $i \in L$$\}$.
Define the (partial) finite projective plane function $\FPP_t : X \cup Y \to \{0, 1\}$ as $\FPP_t(y) = 1 \iff y \in Y$.

We can calculate the 1-sided block sensitivity measures for this function:
\begin{itemize}
\item $\fbs^0(\FPP_t) \geq (t^2+t+1) \cdot \frac{1}{t+1} = \Omega(t)$, as each line gives a sensitive block for $0^n$; since each point belongs to $t+1$ lines, we can assign weight $1/(t+1)$ for each sensitive block and that is a feasible solution for the fractional block sensitivity linear program.
\item $\bs^0(\FPP_t) = 1$, as any two lines intersect, so any two sensitive blocks of $0^n$ overlap.
\item $\bs^1(\FPP_t) = 1$, as there is only one negative input.
\end{itemize}

Next, define $f : S^{\times k} \to \{0, 1\}$ as the composition of $\OR$ with the finite projective plane function, $f = \OR_k (\FPP_t(x^{(1)}), \ldots, \FPP_t(x^{(k)}))$.
By the properties of composition with $\OR$ (see Proposition 31 in \cite{GSS16} for details), we have
\begin{itemize}
\item $\fbs(f) = \max\{\fbs^0(f), \fbs^1(f)\} \geq \fbs^0(f) = \fbs^0(\FPP_t)\cdot k = \Theta(t) \cdot k = \Theta(t\cdot n/t^2) = \Theta(n/t)$,
\item $\bs(f) = \max\{\bs^0(f), \bs^1(f)\} = \bs^0(\FPP_t) \cdot k = k = \Theta(n/t^2)$.
\end{itemize}
As $\sqrt{n \cdot n/t^2} = n/t$, we have $\fbs(f) = \Omega(\sqrt{n\cdot \bs(f)})$ and hence the result. \qedhere
\end{proof}

Note that our example is also tight in regards to the multiplicative constant, since $t$ can be unboundedly large (and the constant arbitrarily close to 1).

\section{Open Ends}

\paragraph{Limitation of the Adversary Bounds.}
In the quantum setting, the certificate barrier shows a limitation on the quantum adversary bounds.
In the classical setting, by our results, fractional block sensitivity characterizes the classical adversary bounds for total functions and thus is of course an upper bound.
Is there a general limitation on the classical adversary methods for partial functions?

\paragraph{Block Sensitivity vs. Fractional Block Sensitivity.}
We have exhibited an example with the largest separation between the two measures for partial functions, $\bs(f) = O(\sqrt{n\cdot\bs(f)})$.
For total functions, one can show that $\fbs(f) \leq \bs(f)^2$, but the best known separations achieve $\fbs(f) = \Omega(\bs(f)^{3/2})$ \cite{GSS16, APV18}.
Can our results be somehow extended for total functions to close the gap?

\section{Acknowledgements}

We are grateful to Rahul Jain for igniting our interest in the classical adversary bounds and Srijita Kundu and Swagato Sanyal for helpful discussions.
We also thank J\={a}nis Iraids for helpful discussions on block sensitivity versus fractional block sensitivity problem.

\bibliography{bibliography}

\newcommand{\etalchar}[1]{$^{#1}$}
\begin{thebibliography}{ABDK16}

\bibitem[Aar06]{Aar06}
Scott Aaronson.
\newblock Lower bounds for local search by quantum arguments.
\newblock {\em SIAM Journal on Computing}, 35(4):804--824, 2006.

\bibitem[Aar08]{Aar08}
Scott Aaronson.
\newblock Quantum certificate complexity.
\newblock {\em Journal of Computer and System Sciences}, 74(3):313--322, 2008.

\bibitem[ABDK16]{ABDK16}
Scott Aaronson, Shalev Ben-David, and Robin Kothari.
\newblock Separations in query complexity using cheat sheets.
\newblock In {\em Proceedings of the Forty-eighth Annual ACM Symposium on
  Theory of Computing}, STOC'16, pages 863--876, New York, NY, USA, 2016. ACM.

\bibitem[AKK16]{AKK16}
Andris Ambainis, Martins Kokainis, and Robin Kothari.
\newblock Nearly optimal separations between communication (or query)
  complexity and partitions.
\newblock In {\em Proceedings of the 31st Conference on Computational
  Complexity}, CCC'16, pages 4:1--4:14. Schloss Dagstuhl--Leibniz-Zentrum fuer
  Informatik, 2016.

\bibitem[Amb00]{Amb00}
Andris Ambainis.
\newblock Quantum lower bounds by quantum arguments.
\newblock In {\em Proceedings of the Thirty-second Annual ACM Symposium on
  Theory of Computing}, STOC'00, pages 636--643, New York, NY, USA, 2000. ACM.

\bibitem[Amb03]{Amb03}
Andris Ambainis.
\newblock Polynomial degree vs. quantum query complexity.
\newblock In {\em Proceedings of the 44th Annual IEEE Symposium on Foundations
  of Computer Science}, FOCS'03, pages 230--239, Washington, DC, USA, 2003.
  IEEE Computer Society.

\bibitem[APV18]{APV18}
Andris Ambainis, Kri\v{s}j\={a}nis Pr\={u}sis, and Jevg\={e}nijs Vihrovs.
\newblock On block sensitivity and fractional block sensitivity.
\newblock {\em Lobachevskii Journal of Mathematics}, 39(7):967--–969, 2018.

\bibitem[BBC{\etalchar{+}}01]{BBCMdW01}
Robert Beals, Harry Buhrman, Richard Cleve, Michele Mosca, and Ronald de~Wolf.
\newblock Quantum lower bounds by polynomials.
\newblock {\em Journal of the ACM}, 48(4):778--797, 2001.

\bibitem[BDK16]{BDK16}
Shalev Ben-David and Robin Kothari.
\newblock Randomized query complexity of sabotaged and composed functions.
\newblock In {\em 43rd International Colloquium on Automata, Languages, and
  Programming (ICALP 2016)}, volume~55 of {\em Leibniz International
  Proceedings in Informatics (LIPIcs)}, pages 60:1--60:14, Dagstuhl, Germany,
  2016. Schloss Dagstuhl--Leibniz-Zentrum fuer Informatik.

\bibitem[BdW02]{BdW02}
Harry Buhrman and Ronald de~Wolf.
\newblock Complexity measures and decision tree complexity: a survey.
\newblock {\em Theoretical Computer Science}, 288(1):21 -- 43, 2002.

\bibitem[GSS16]{GSS16}
Justin Gilmer, Michael Saks, and Srikanth Srinivasan.
\newblock Composition limits and separating examples for some {B}oolean
  function complexity measures.
\newblock {\em Combinatorica}, 36(3):265--311, 2016.

\bibitem[HLS07]{HTS07}
Peter Hoyer, Troy Lee, and Robert Spalek.
\newblock Negative weights make adversaries stronger.
\newblock In {\em Proceedings of the Thirty-ninth Annual ACM Symposium on
  Theory of Computing}, STOC'07, pages 526--535, New York, NY, USA, 2007. ACM.

\bibitem[HNS01]{HNS01}
Peter H{\o}yer, Jan Neerbek, and Yaoyun Shi.
\newblock Quantum complexities of ordered searching, sorting, and element
  distinctness.
\newblock In {\em Proceedings of the 28th International Colloquium on Automata,
  Languages and Programming,}, ICALP'01, pages 346--357, London, UK, UK, 2001.
  Springer-Verlag.

\bibitem[JK10]{JK10}
Rahul Jain and Hartmut Klauck.
\newblock The partition bound for classical communication complexity and query
  complexity.
\newblock In {\em Proceedings of the 2010 IEEE 25th Annual Conference on
  Computational Complexity}, CCC'10, pages 247--258, Washington, DC, USA, 2010.
  IEEE Computer Society.

\bibitem[KT16]{KT16}
Raghav Kulkarni and Avishay Tal.
\newblock On fractional block sensitivity.
\newblock {\em Chicago Journal Of Theoretical Computer Science}, 8:1--16, 2016.

\bibitem[LM04]{LM04}
Sophie Laplante and Frédéric Magniez.
\newblock Lower bounds for randomized and quantum query complexity using
  {K}olmogorov arguments.
\newblock In {\em Proceedings of the 19th IEEE Annual Conference on
  Computational Complexity}, CCC'04, pages 294--304, Washington, DC, USA, 2004.
  IEEE Computer Society.

\bibitem[LV08]{MV08}
Ming Li and Paul~M.B. Vitnyi.
\newblock {\em An Introduction to {K}olmogorov Complexity and Its
  Applications}.
\newblock Springer Publishing Company, Incorporated, 3 edition, 2008.

\bibitem[Nis89]{Nis89}
Noam Nisan.
\newblock {CREW PRAM}s and decision trees.
\newblock In {\em Proceedings of the Twenty-first Annual ACM Symposium on
  Theory of Computing}, STOC'89, pages 327--335, New York, NY, USA, 1989. ACM.

\bibitem[NS94]{NS94}
Noam Nisan and Mario Szegedy.
\newblock On the degree of {B}oolean functions as real polynomials.
\newblock {\em Computational Complexity}, 4(4):301--313, 1994.

\bibitem[Rei09]{Rei09}
Ben~W. Reichardt.
\newblock Span programs and quantum query complexity: The general adversary
  bound is nearly tight for every {B}oolean function.
\newblock In {\em Proceedings of the 2009 50th Annual IEEE Symposium on
  Foundations of Computer Science}, FOCS'09, pages 544--551, Washington, DC,
  USA, 2009. IEEE Computer Society.

\bibitem[{\v S}S06]{SS06}
Robert {\v S}palek and Mario Szegedy.
\newblock All quantum adversary methods are equivalent.
\newblock {\em Theory of Computing}, 2(1):1--18, 2006.

\bibitem[Tal13]{Tal13}
Avishay Tal.
\newblock Properties and applications of {B}oolean function composition.
\newblock In {\em Proceedings of the 4th Conference on Innovations in
  Theoretical Computer Science}, ITCS'13, pages 441--454, New York, NY, USA,
  2013. ACM.

\bibitem[Yao77]{Yao77}
Andrew Chi-Chin Yao.
\newblock Probabilistic computations: Toward a unified measure of complexity.
\newblock In {\em 18th Annual Symposium on Foundations of Computer Science},
  pages 222--227, 1977.

\end{thebibliography}

\newpage

\appendix

\section{Rank-1 Relational Adversary Definition} \label{caa-def}

\begin{proof}[Proof of Proposition \ref{caa-prop}]
Let $u, v$ be vectors that maximize $\CA_1(f)$.
Let $h \in H$ be any letter and $S_h = f^{-1}(h)$.
Since for every $x, y$, such that $f(x) = f(y)$, we have $u(x)v(y) = 0$, it follows that either $u(x) = 0$ for all $x \in S_h$ or $v(x) = 0$ for all $x \in S_h$.
Therefore, we can find a partition $A \cup B = H$ such that:
\begin{itemize}
\item if $u(x) > 0$, then $f(x) \in A$;
\item if $v(x) > 0$, then $f(y) \in B$;
\item for every $h \in H$, either $h \in A$ or $h \in B$.
\end{itemize}
This partition therefore also defines a partition of the inputs, $X \cup Y = S$, where $X = f^{-1}(A)$ and $Y = f^{-1}(B)$.

Now, notice that $\theta(x,i) $ does not depend on the particular choice of $ x $ if $ x_i:=g_1 \in G$ is fixed. 
Similarly, let $ y_i := g_2  \in G$ be fixed, then $ \theta(y,i) $ does not depend on the particular choice of $y $. 
This allows to simplify the expression for $ \CA_1(f) $, since for each $ i $ we can fix values $ g_1 \neq g_2 $ (such that there exist  $ x \in X$, $ y \in Y$ with $ u(x) v(y)>0 $ and $ x_i  =g_1 $ and $ y=g_2 $) and ignore the remaining components of  $ x $, $ y $, i.e.,
\begin{equation*}
\CA_1(f) = \max_{A,B: A\cup B = H} \max_{u, v} \min_{\substack{i \in [n], \\ g_1, g_2 \in G, g_1 \neq g_2: \\ \exists x \in X, y \in Y: \\ x_i = g_1, y_i = g_2,\\  u(x)v(y) > 0}} \max\left\{\frac{\sum_{y \in Y} v(y)}{\sum_{y \in Y : y_i \neq g_1} v(y)}, \frac{\sum_{x \in X} u(x)}{\sum_{x \in X : x_i \neq g_2} u(x)}\right\}.
\end{equation*}

Further assume that both $X$ and $Y$ are non-empty, because otherwise the value of $\CA_1$ would not be defined.
Notice that multiplying either $ u $ or $ v $ with any scalar does not affect the value of $ \CA_1 $.
Hence, we can scale $u$ and $v$ to probability distributions $p$ and $q$ over $X$ and $Y$, respectively.
More specifically, we can further simplify $\CA_1$:
\begin{align*}
\CA_1(f) &= \max_{A,B:\atop A\cup B = H} \max_{p, q} \min_{\substack{i \in [n], \\ g_1, g_2 \in G: g_1 \neq g_2: \\ \exists x \in X, y \in Y: \\ x_i = g_1, y_i = g_2,\\  p(x)q(y) > 0}} \frac{1}{\min\left\{\sum_{y \in Y \atop y_i \neq g_1} q(y), \sum_{x \in X \atop x_i \neq g_2} p(x)\right\}} \\
&= \max_{A,B:\atop A\cup B = H} \max_{p, q} \min_{\substack{i \in [n], \\ g_1, g_2 \in G: g_1 \neq g_2: \\ \exists x \in X, y \in Y: \\ x_i = g_1, y_i = g_2,\\  p(x)q(y) > 0}} \frac{1}{\min\left\{\Pr_{y \sim q}[y_i \neq g_1], \Pr_{x \sim p}[x_i \neq g_2]\right\}}.
\end{align*}
\end{proof}

We can further simplify this definition if the inputs are Boolean:
\begin{proposition} \label{bool-caa}
Let $f : S \to H$, where $S \subseteq \{0, 1\}^n$.
Let $A \cup B = H$ be a partition of the output alphabet, i.e., $A \cap B = \varnothing$.
Let $p$ and $q$ be probability distributions over $X:=f^{-1}(A)$ and $Y:=f^{-1}(B)$, respectively.
Then
\begin{equation*}
\CA_1(f) = \max_{A,B, \atop p,q} \min_{i \in [n], \atop b \in \{0, 1\}} \frac{1}{\min\left\{\Pr_{y \sim q}[y_i \neq b], \Pr_{x \sim p}[x_i = b]\right\}}.
\end{equation*}
\end{proposition}

\begin{proof}
For $g_1, g_2 \in \{0, 1\}$, $ g_1 \neq g_2 $ implies $ g_2 = g_1 \oplus 1 $. It follows that
  \[ 
  \CA_1 (f) = 
  \max_{A, B, \atop p, q}  
 \min_{\substack{i \in [n], \\ b  \in \{0,1\}:\\    \exists x \in X, y \in Y:\\ x_i = b,  y_i \neq  b,\\ p(x)q(y) > 0} }
\frac{1}{\min\left\{\Pr_{y \sim q}[y_i \neq b], \Pr_{x \sim p}[x_i = b]\right\}}
.
 \]
 Moreover, we can drop the requirement $\exists x \in X, y \in Y: x_i = b,  y_i \neq  b, p(x)q(y) > 0$. To see that, fix any $p, q$,  and consider the quantities
 \begin{align*}
 &	\alpha =  
 \max_{i \in [n]} 
 	\max_{\substack{b  \in \{0,1\}:\\    \exists x \in X, y \in Y:\\ x_i = b,  y_i \neq  b,\\ p(x)q(y) > 0} }
 	\min \left\{\Pr_{x \sim p}[x_i = b], \Pr_{y \sim q}[y_i \neq b]\right\}
 	 \\
 & \beta =  
 \max_{i \in [n] \atop b  \in \{0,1\} }
 \min \left\{\Pr_{x \sim p}[x_i = b], \Pr_{y \sim q}[y_i \neq b]\right\}.
 \end{align*}
 Clearly, $ \alpha \leq  \beta  $. To show the converse inequality, consider any $ i\in [n] $ and (if such exists) $ b \in \{0,1\} $ satisfying
$  u(x) v(y) = 0 $ for any $x \in X, y \in Y$ with $ x_i = b $, $ y_i \neq b $ (to deal with the possibility no such $x, y$ exist, we consider the empty sum to be zero). Then also
\[ 
0=
\sum_{\substack{x \in X, y \in Y \\ x_i = b,\ y_i \neq b} }
p(x) q(y)=
\left(\sum_{x \in X :  x_i = b  }p(x) \right)\left(\sum_{y \in Y :  y_i \neq  b  }q(y)\right) = \Pr_{x \sim p}[x_i = b] \cdot \Pr_{y \sim q}[y_i \neq b].
 \]
 Therefore,  $  \min \left\{\Pr_{x \sim p}[x_i = b] , \Pr_{y \sim q}[y_i \neq b]
 \right\} =0  \leq  \alpha$. Thus $ \alpha = \beta $. Thus the claim follows.
\end{proof}

We also note that $ \CA_1(f)  $ can be found the following way. 
Let $A \cup B = H$ be any  suitable partition of $ H $ and denote
\begin{equation*}
\CA_1(f,A,B) =
\max_{p, q} \min_{\substack{i \in [n], \\ g_1, g_2 \in G: g_1 \neq g_2: \\ \exists x \in X, y \in Y: \\ x_i = g_1, y_i = g_2,\\  p(x)q(y) > 0}} \frac{1}{\min\left\{\Pr_{y \sim q}[y_i \neq g_1], \Pr_{x \sim p}[x_i \neq g_2]\right\}}.
\end{equation*}
Then $ \CA_1(f) = \max_{A,B} \CA_1(f,A,B) $. On the other hand, for each fixed partition $ A,B $ the value $ \CA_1(f,A,B) $ can be found from the following program:
\begin{proposition}
Let $f : S \to H$, where $S \subseteq G^n$.
Let $A \cup B = H$ be any partition of $H$ such that $A, B \neq \varnothing$.
Let $X = f^{-1}(A)$ and $Y = f^{-1}(B)$.
The value of $\CA_1(f,A,B)$ is equal to the optimal solution of the following program:
\begin{align*}
\text{maximize} \sum_{x \in X} w_x \hspace{1cm} \text{s.t.} \hspace{.5cm} 
&\sum_{x \in X} w_y  = \sum_{y \in Y} w_y  ,   \\
&\min\left \{\sum_{\substack{x \in X : \\ x_i \neq g_2}} w_x ,  \sum_{\substack{y \in Y : \\ y_i \neq  g_1}} w_y 
\right\} \leq 1, &   \substack{i \in [n], \\ g_1, g_2 \in G , g_1 \neq g_2, \\ w_xw_y > 0} \\
&w_x \geq 0,  &   x \in S.
\end{align*}
\end{proposition}

The proof is analogous to that of Lemma \ref{cmmAlt}.
\begin{proof}
Denote the optimal value of this program by $ \mu $. Then $ \mu \leq \CA_1 (f,A,B) $, since we can take $p(x) = w_x / \mu  $, $ q(y) = w_y/\mu  $ (where $\{w_x\}_{x \in S} $) is an optimal solution of the program).
This way we obtain a feasible solution for $ \CA_1(f) $, which gives 
\[ 
\min \left\{
	{\sum_{\substack{x \in X : \\ x_i \neq g_2}} p(x) },
	{\sum_{\substack{y \in Y : \\ y_i \neq g_1}} q(y) }
\right\}
=
\frac{1}{\mu} \cdot
\min \left\{\sum_{\substack{x \in X : \\ x_i \neq g_2}} w_x ,  \sum_{\substack{y \in Y : \\ y_i \neq  g_1}} w_y \right\} \leq \frac 1 \mu
 \]
for each $ i \in [n], g_1, g_2 \in G$ such that $g_1 \neq g_2$ and there exist $x \in X, y \in Y$ with $x_i = g_1$ and $y_i = g_2$, thus $ \CA_1(f,A,B) \geq \mu $.

Let us show the converse inequality.
If the probability distributions $ p,q $ provide an optimal solution for $ \CA_1(f,A,B)  $, then 	$ w_x = p(x) \cdot \CA_1(f,A,B) $ and $ w_y = q(y) \cdot \CA_1(f,A,B)  $ gives a feasible solution for the program and the value  this solution is $ \sum_{x \in X} w_x =\CA_1(f,A,B) $.
 Hence, also $ \CA_1(f,A,B)\leq \mu $.
\end{proof}

 For Boolean outputs, the partition of $H$ can be fixed to $A=\{0\}, B= \{1\}$, giving a single program. For Boolean inputs, the condition $g_1, g_2 \in G, g_1 \neq g_2, w_xw_y > 0$ can be replaced simply by $b \in \{0, 1\}$ by Proposition \ref{bool-caa}.
Therefore, for Boolean functions this program can be recast as a mixed-integers linear program, providing an algorithm for finding $ \CA_1(f) $.
\end{document}